\documentclass[aps,prl,twocolumn,superscriptaddress,showpacs,10pt]{revtex4}

\usepackage[T1]{fontenc}
\usepackage{times}
\usepackage{color,graphicx}
\usepackage{array}
\usepackage{amsmath}
\usepackage{amssymb}
\usepackage{amsthm}

\newcommand{\ket}[1]{| #1 \rangle}

\newcommand{\ketbra}[2]{| #1 \rangle\langle #2 |}
\newcommand{\bb}[1]{\mathbb{#1}}
\newcommand{\cl}[1]{\mathcal{#1}}

\newcommand\Tr{\mathop{\rm Tr}\nolimits}

\newtheorem{theorem}{Theorem}

\makeatletter
\def\Ddots{\mathinner{\mkern1mu\raise\p@
\vbox{\kern7\p@\hbox{.}}\mkern2mu
\raise4\p@\hbox{.}\mkern2mu\raise7\p@\hbox{.}\mkern1mu}}
\makeatother

\begin{document}


\title{Non-Positive Partial Transpose Subspaces Can be as Large as Any Entangled Subspace}

\author{Nathaniel Johnston}%
\affiliation{Institute for Quantum Computing, University of Waterloo, Waterloo, Ontario, Canada}%

\begin{abstract}
	It is known that, in an $(m \otimes n)$-dimensional quantum system, the maximum dimension of a subspace that contains only entangled states is $(m-1)(n-1)$. We show that the exact same bound is tight if we require the stronger condition that every state with range in the subspace has non-positive partial transpose. As an immediate corollary of our result, we solve an open question that asks for the maximum number of negative eigenvalues of the partial transpose of a quantum state. In particular, we give an explicit method of construction of a bipartite state whose partial transpose has $(m-1)(n-1)$ negative eigenvalues, which is necessarily maximal, despite recent numerical evidence that suggested such states may not exist for large $m$ and $n$.
\end{abstract}



\maketitle

In quantum information theory, a pure state $\ket{v} \in \mathbb{C}^m \otimes \mathbb{C}^n$ is called a \emph{product state} if it can be written in the form $\ket{v} = \ket{v_1} \otimes \ket{v_2}$ for some $\ket{v_1} \in \mathbb{C}^m$ and $\ket{v_2} \in \mathbb{C}^n$; otherwise it is called \emph{entangled}. Similarly, a mixed state $\rho \in M_m \otimes M_n$ is called \emph{separable} if it can be written in the form $\rho = \sum_i p_i \ketbra{v_i}{v_i}$ for some real constants $p_i > 0$ with $\sum_i p_i = 1$ and product states $\ket{v_i}$; otherwise it is called \emph{entangled}.

The problem of determining whether or not a given mixed state is entangled is one of the central questions in quantum information theory, and it is expected that no efficient procedure for answering this question in full generality exists \cite{G03,G10}. However, there are many known one-sided tests that can be used to prove that a given state is entangled. The most well-known such test is the \emph{positive partial transpose (PPT)} criterion \cite{P96}, which says that if $\rho$ is separable then $\rho^\Gamma$ is positive semidefinite, where $\Gamma$ refers to the linear \emph{partial transposition} map that sends $\ketbra{i}{j} \otimes \ketbra{k}{\ell}$ to $\ketbra{i}{j} \otimes \ketbra{\ell}{k}$ (i.e., it is the map $id_m \otimes T : M_m \otimes M_n \rightarrow M_m \otimes M_n$, where $id_m$ is the identity map and $T$ is the usual transpose map with respect to the standard basis $\{\ket{i}\}$). If $\rho^\Gamma$ is positive semidefinite, it is said that $\rho$ is \emph{positive partial transpose (PPT)}; otherwise, it is called \emph{non-positive partial transpose (NPT)}.

It is known that $(m-1)(n-1)$ is the maximum dimension of a subspace $\cl{S} \subseteq \mathbb{C}^m \otimes \mathbb{C}^n$ such that every $\ket{v} \in \cl{S}$ is entangled \cite{Wal02,Par04}. In the present paper, we consider the related problem of finding the maximum dimension of a subspace $\cl{S}$ such that every $\rho \in M_m \otimes M_n$ with range contained in $\cl{S}$ is NPT. Since all NPT states are entangled, it follows immediately that no such subspace of dimension greater than $(m-1)(n-1)$ exists. Our main result shows that this bound is tight for all $m$ and $n$. That is, there exists a subspace of dimension $(m-1)(n-1)$ that is not only entangled, but is even NPT. This result is perhaps surprising, since most results concerning the relationship between the PPT criterion and separability show that these two properties become more distant from each other as $m$ and $n$ increase. For example, the converse of the PPT criterion (that is, the statement that if $\rho$ is entangled then it is NPT) only holds when $mn \leq 6$ \cite{HHH96}, the set of PPT states is much larger than the set of separable states in general \cite{AS06,Ye09}, and PPT states can be very far from the set of separable states when $m$ and $n$ are large \cite{BS10}.

As an important consequence of our result, we resolve completely the question of how many negative eigenvalues the partial transpose of a bipartite mixed state can have. This question is motivated by the facts that some measures of entanglement are defined in terms of the negative eigenvalues of $\rho^\Gamma$ \cite{ZHSL98,VW02} and that bounds on the number of negative eigenvalues of $\rho^\Gamma$ have recently been used to show that squared negativity is not a lower bound of geometric discord \cite{RP12}. It is known \cite{Ran13} that the partial transpose of a state can not have more than $(m-1)(n-1)$ negative eigenvalues. However, tightness has only been shown when $\min\{m,n\} \leq 2$ \cite{CD12} or $m = n = 3$, and recent numerical evidence \cite{JJCC08,Ran13} suggested that this bound may not be tight for larger values of $m$ and $n$. We show, via explicit construction, that the contrary is true; for all $m,n$ there exists $\rho \in M_m \otimes M_n$ such that $\rho^\Gamma$ has $(m-1)(n-1)$ negative eigenvalues.

We now present our main result, which shows that, for any entangled subspace, there exists an NPT subspace of the same dimension. Our proof is by explicit construction, and builds upon the ideas presented in the proof of \cite[Proposition~10]{CW08}.
\begin{theorem}\label{thm:nppt_subspace}
	There exists a subspace $\mathcal{S} \subseteq \mathbb{C}^m \otimes \mathbb{C}^n$ of dimension $(m-1)(n-1)$ such that every density matrix with range contained in $\mathcal{S}$ is NPT.
\end{theorem}
\begin{proof}
	We begin by defining the subspace that we will prove has the desired property:
	\begin{align*}
		\mathcal{S} := {\rm span}\big\{ \ket{j}\ket{k+1} - \ket{j+1}\ket{k} : \ & 0 \leq j \leq m-2, \\
		                                                                        & 0 \leq k \leq n-2 \big\}.
	\end{align*}
	It is clear that $\cl{S}$ has dimension $(m-1)(n-1)$, so we now focus on the NPT condition. To this end, we first let $\Delta : \bb{C}^m \otimes \bb{C}^n \rightarrow M_{n,m}$ be the linear isomorphism that sends $\ket{i} \otimes \ket{j}$ to the $n \times m$ matrix $\ketbra{j}{i}$. Then $\Delta(\cl{S})$ is easily seen to equal the set of $n \times m$ matrices with the property that each of their $n + m - 1$ anti-diagonals sum to $0$.
	
	Now let $\rho \in M_m \otimes M_n$ be such that its range is contained in $\mathcal{S}$. Then we can write
	\begin{align*}
		\rho = \sum_i p_i\ketbra{v_i}{v_i}
	\end{align*}
	for some constants $p_i > 0$ with $\sum_i p_1 = 1$ and some pure states $\ket{v_i} \in \mathcal{S}$. Let $j_0<j_1,k_0<k_1$ be such that, for all $i$, the $2 \times 2$ submatrix of $\Delta(\ket{v_i})$ corresponding to rows $\ket{j_0}$ and $\ket{j_1}$ and columns $\ket{k_0}$ and $\ket{k_1}$ is of the form
	\begin{align*}
		\begin{bmatrix}0 & b_i \\ a_i & c_i\end{bmatrix},
	\end{align*}
	where the constants $\{a_i\}$ and $\{b_i\}$ satisfy $\sum_i p_i\overline{a_i}b_i \neq 0$ (the existence of such a $2 \times 2$ submatrix is not obvious -- we defer the proof of its existence to the end of the proof of this theorem, since it is slightly technical).
	
	A direct calculation now reveals that the $2 \times 2$ principle submatrix of $(\ketbra{v_i}{v_i})^\Gamma$ corresponding to rows and columns $\ket{j_0}\ket{k_0}$ and $\ket{j_1}\ket{k_1}$ is
	\begin{align*}
		\begin{bmatrix}0 & a_i\overline{b_i} \\ \overline{a_i}b_i & |c_i|^2\end{bmatrix},
	\end{align*}
	so the same principle submatrix of $\rho^\Gamma$ is
	\begin{align*}
		\sum_i p_i \begin{bmatrix}0 & a_i\overline{b_i} \\ \overline{a_i}b_i & |c_i|^2\end{bmatrix}.
	\end{align*}
	Since the determinant of this principle submatrix is
	\begin{align*}
		-\Big| \sum_i p_i\overline{a_i}b_i \Big|^2 < 0,
	\end{align*}
	it follows that it has a negative eigenvalue, so $\rho^\Gamma$ has a negative eigenvalue as well, as desired.
	
	All that remains is to prove that there exist $j_0<j_1,k_0<k_1$ such that the $2 \times 2$ submatrix of $\Delta(\ket{v_i})$ corresponding to rows $\ket{j_0},\ket{j_1}$ and columns $\ket{k_0},\ket{k_1}$ equals
	\begin{align*}
		\begin{bmatrix}0 & b_i \\ a_i & c_i\end{bmatrix}
	\end{align*}
	and $\sum_i p_i\overline{a_i}b_i \neq 0$. It follows from the construction of $\cl{S}$ that no anti-diagonal of $\Delta(\ket{v_i})$ has exactly $1$ non-zero entry (they all have either $0$ or $2$ or more non-zero entries). It is then clear that a $2 \times 2$ submatrix can be found with top-left entry equal to $0$ for all $i$ -- simply choose the leftmost anti-diagonal that is non-zero in at least one of the $\Delta(\ket{v_i})$'s and choose any $2 \times 2$ submatrix containing $2$ non-zero entries on that anti-diagonal.
	
	This submatrix will, by construction, have $\overline{a_i}b_i \neq 0$ for at least one choice of $i$. However, it could still happen that $\sum_i p_i\overline{a_i}b_i = 0$ by terms negating each other in the summation. To see that this problem can always be avoided, write
	\begin{align*}
		\Delta(\ket{v_i}) = \begin{bmatrix}0 & \cdots & 0 & 0 & \cdots & 0 & d_{i,L} & * & \cdots & * \\
		0 & \cdots & 0 & 0 & \cdots & d_{i,L-1} & * & * & \cdots & * \\ 
		\vdots & \Ddots & \vdots & \vdots & \Ddots & \vdots & \vdots & \vdots & \Ddots & \vdots \\ 
		0 & \cdots & 0 & d_{i,2} & \cdots & * & * & * & \cdots & * \\ 
		0 & \cdots & d_{i,1} & * & \cdots & * & * & * & \cdots & *\end{bmatrix},
	\end{align*}
	where at least one value of $d_{i,j}$ is non-zero and $*$ indicates an entry whose value is irrelevant to us. That is, we define $L$ to be the length of the leftmost anti-diagonal that is non-zero in at least one $\Delta(\ket{v_i})$, and we define $d_{i,1},d_{i,2},\ldots,d_{i,L}$ to be the entries of this anti-diagonal. Let $j_0$ be the smallest integer for which $d_{i,j_0} \neq 0$ for some $i$ and define $a_i := d_{i,j_0}$ for all $i$ (i.e., the bottom-left corner of the $2 \times 2$ submatrix that we are choosing is the lowest-left entry in this anti-diagonal that is non-zero for some $i$).
	
	By using the fact that each anti-diagonal of $\Delta(\ket{v_i})$ sums to $0$, we have
	\begin{align*}
		d_{i,L} = -\sum_{j=j_0}^{L-1} d_{i,j} \quad \forall \, i.
	\end{align*}
	Now suppose for a contradiction that $\sum_i p_i\overline{a_i}d_{i,j} = 0$ for all $j > j_0$. By using this assumption twice, we see that
	\begin{align*}
		0 & = \sum_i p_i\overline{a_i}d_{i,L} \\
		  & = \sum_i p_i\overline{a_i}\Big(-\sum_{j=j_0}^{L-1} d_{i,j}\Big) \\
		  & = -\sum_{j=j_0}^{L-1} \Big(\sum_i p_i\overline{a_i}d_{i,j}\Big) \\
		  & = -\sum_i p_i|a_i|^2 - \sum_{j=j_0+1}^{L-1} \Big(\sum_i p_i\overline{a_i}d_{i,j}\Big) \\
		  & = -\sum_i p_i|a_i|^2 \\
		  & < 0,
	\end{align*}
	which is the desired contradiction. It follows that there exists $j_1 > j_0$ such that $\sum_i p_i\overline{a_i}d_{i,j_1} \neq 0$, so we define $b_i := d_{i,j_1}$ for all $i$ (i.e., we choose the top-right corner of the $2 \times 2$ submatrix to be the $j_1$-th entry of the anti-diagonal we are working with). Since we have found a $2 \times 2$ submatrix for which $\sum_i p_i\overline{a_i}b_i \neq 0$, the proof is complete.
\end{proof}

We now turn our attention to the number of negative eigenvalues of the partial transpose of a state. It has been shown that for all $\rho \in M_m \otimes M_n$, $\rho^\Gamma$ can not have more than $(m-1)(n-1)$ negative eigenvalues \cite{Ran13} (also see \cite[Corollary~5.4]{JK10} for an earlier, but less direct proof). However, tightness of this bound has not been known -- for example, when $m = 3$ and $n = 4$, numerical evidence was presented in \cite{Ran13} that suggested that the maximum number of negative eigenvalues of $\rho^\Gamma$ might be $5$, not $(m-1)(n-1) = 6$.

We now use Theorem~\ref{thm:nppt_subspace} to show that the bound $(m-1)(n-1)$ is in fact tight for all $m$ and $n$. For the most part, the proof is elementary, but it does require some familiarity with \emph{dual cones}. Given a set of Hermitian matrices $\cl{C} \subseteq M_m \otimes M_n$, the dual cone of $\cl{C}$ is
\begin{align*}
	\cl{C}^\circ := \big\{ Y = Y^\dagger \in M_m \otimes M_n : \Tr(XY) \geq 0 \ \ \forall \, X \in \cl{C} \big\}.
\end{align*}
In particular, the set of positive semidefinite matrices is its own dual cone. More generally, dual cones can be defined on any real Hilbert space, but the definition given here suffices for our purposes. For basic properties of dual cones, the reader is directed to \cite{BV04}.
\begin{theorem}\label{thm:neg_evals}
	For all $m$ and $n$, there exists $\rho \in M_m \otimes M_n$ such that $\rho^\Gamma$ has $(m-1)(n-1)$ negative eigenvalues.
\end{theorem}
\begin{proof}
	To prove the statement, we explicitly construct $\rho \in M_m \otimes M_n$ with the desired property. Let $P \in M_m \otimes M_n$ be the orthogonal projection onto the $(m-1)(n-1)$-dimensional NPT subspace described by Theorem~\ref{thm:nppt_subspace}. Since there is no PPT state $\sigma$ with $P\sigma P = \sigma$, we have $\Tr(P\sigma) < 1$ for all PPT $\sigma$. Furthermore, since the set of PPT states is compact, there exists a real constant $0 < c < 1$ such that $\Tr(P\sigma) \leq c$ for all PPT $\sigma$. If we define the operator $X := I - \tfrac{1}{c}P$ then it is easily verified that $X$ has $(m-1)(n-1)$ negative eigenvalues and $\Tr(X\sigma) \geq 0$ for all PPT $\sigma$. The latter fact is equivalent to the statement that $X$ is in the dual cone of the set of PPT states. This dual cone is easily seen to equal
	\begin{align*}
		\big\{ Y_1 + Y_2^\Gamma : Y_1,Y_2 \in M_m \otimes M_n \text{ are positive semidefinite} \big\}
	\end{align*}
	(an explicit proof of this fact is given by \cite[Corollary~3.7]{SSZ09}). Thus there exist positive semidefinite $X_1,X_2$ such that $X = X_1 + X_2^\Gamma$. Finally, we define $\rho := X_2/\Tr(X_2)$, which is a valid density matrix by definition. Furthermore, since $X$ has $(m-1)(n-1)$ negative eigenvalues, and (up to scaling) $\rho^\Gamma = X - X_1$, it follows that $\rho^\Gamma$ has at least $(m-1)(n-1)$ negative eigenvalues as well (and hence \emph{exactly} $(m-1)(n-1)$ negative eigenvalues), which completes the proof.
\end{proof}

We note that the procedure to construct $\rho$ described in the proof of Theorem~\ref{thm:neg_evals} is in fact completely constructive and can be carried out efficiently. Indeed, the quantity $c$ and operators $X_1$ and $X_2$ can be found via semidefinite programming, and there are known efficient methods for solving semidefinite programs \cite{GLS93}. For an introduction to semidefinite programming from the perspective of quantum information theory, the reader is directed to \cite{Wat04Lec7}. A semidefinite program that finds $\rho \in M_m \otimes M_n$ such that $\rho^\Gamma$ has $(m-1)(n-1)$ negative eigenvalues is as follows, where $P$ is the orthogonal projection onto the NPT subspace described by Theorem~\ref{thm:nppt_subspace}, and we optimize over $d \in \bb{R}$ and density matrices $\rho \in M_m \otimes M_n$:
\begin{align*}
\begin{matrix}
\begin{tabular}{r l}
\text{maximize:} & $d$ \\
\text{subject to:} & $\rho^\Gamma \leq I - dP$ \\
\ & $\Tr(\rho) = 1$ \\
\ & $\rho \geq 0$ \\
\end{tabular}
\end{matrix}
\end{align*}
It is straightforward to see that any feasible point $\rho$ corresponding to a value of $d > 1$ is such that $\rho^\Gamma$ has $(m-1)(n-1)$ negative eigenvalues, as desired -- the proof of Theorem~\ref{thm:neg_evals} demonstrates that such a value of $d$ exists, since we can choose $d = 1/c$. As an explicit example that arises from making use of this semidefinite program, we now present in the standard basis a density matrix $\rho \in M_3 \otimes M_4$ such that $\rho^\Gamma$ has $6$ negative eigenvalues, beating the best known lower bound of $5$ (we use $\cdot$ to indicate $0$ entries):
\begin{align*}
	\rho := \frac{1}{34}\left[\begin{array}{cccc|cccc|cccc}
	  9     & \cdot & \cdot & \cdot & \cdot & 3     & \cdot & \cdot & \cdot & \cdot & 1     & \cdot \\
		\cdot & 3     & \cdot & \cdot & \cdot & \cdot & 2     & \cdot & \cdot & \cdot & \cdot & 1     \\
    \cdot & \cdot & 1     & \cdot & \cdot & \cdot & \cdot & 1     & \cdot & \cdot & \cdot & \cdot \\
    \cdot & \cdot & \cdot & \cdot & \cdot & \cdot & \cdot & \cdot & \cdot & \cdot & \cdot & \cdot \\ \hline
		\cdot & \cdot & \cdot & \cdot & 2     & \cdot & \cdot & \cdot & \cdot & 1     & \cdot & \cdot \\
    3     & \cdot & \cdot & \cdot & \cdot & 2     & \cdot & \cdot & \cdot & \cdot & 2     & \cdot \\
    \cdot & 2     & \cdot & \cdot & \cdot & \cdot & 2     & \cdot & \cdot & \cdot & \cdot & 3     \\
    \cdot & \cdot & 1     & \cdot & \cdot & \cdot & \cdot & 2     & \cdot & \cdot & \cdot & \cdot \\ \hline
		\cdot & \cdot & \cdot & \cdot & \cdot & \cdot & \cdot & \cdot & \cdot & \cdot & \cdot & \cdot \\
    \cdot & \cdot & \cdot & \cdot & 1     & \cdot & \cdot & \cdot & \cdot & 1     & \cdot & \cdot \\
    1     & \cdot & \cdot & \cdot & \cdot & 2     & \cdot & \cdot & \cdot & \cdot & 3     & \cdot \\
    \cdot & 1     & \cdot & \cdot & \cdot & \cdot & 3     & \cdot & \cdot & \cdot & \cdot & 9     \end{array}\right].
\end{align*}
It is straightforward to verify that $\rho$ is positive semidefinite, yet $\rho^\Gamma$ has negative eigenvalues equal to approximately $-0.0204$, $-0.0159$, and $-0.0105$, each with multiplicity $2$. More generally, a MATLAB script that uses the CVX package \cite{cvx} to solve this semidefinite program, and thus constructs $\rho \in M_m \otimes M_n$ such that $\rho^\Gamma$ has $(m-1)(n-1)$ negative eigenvalues, can be downloaded from \cite{JohNegEvalCode}. To give a rough idea of the speed of this script, the above density matrix in $M_3 \otimes M_4$ was computed in about $0.6$ seconds on a standard desktop computer, while a density matrix $\rho \in M_{13} \otimes M_{13}$ whose partial transpose has $144$ negative eigenvalues takes about $30$ minutes to compute.

In this brief article, we answered the question of how large a subspace $\cl{S} \subseteq \bb{C}^m \otimes \bb{C}^n$ can be such that every density matrix $\rho \in M_m \otimes M_n$ with range contained in $\cl{S}$ is NPT. More specifically, we have shown that such subspaces can have dimension $(m-1)(n-1)$, which is just as large as the largest subspace consisting entirely of entangled states. We then used this result to resolve a long-standing question that asks for the maximum number of negative eigenvalues that the partial transpose of a state $\rho \in M_m \otimes M_n$ can have -- the answer to this question is also $(m-1)(n-1)$.

As a possible extension to this work, it may be worth investigating the number of negative eigenvalues of $(id_m \otimes \Phi)(\rho)$, where $\Phi : M_n \rightarrow M_n$ is a given positive but not completely positive map. It follows from \cite[Corollary~5.4]{JK10} that $(id_m \otimes \Phi)(\rho)$ can not have more than $(m-1)(n-1)$ negative eigenvalues (and more generally, if $\Phi$ is $k$-positive then $(id_m \otimes \Phi)(\rho)$ can not have more than $(m-k)(n-k)$ negative eigenvalues), but for which maps is this bound tight? We have shown that the transpose map is one such example, but are there others of interest?

In another direction, it would be interesting to extend our results to the multipartite setting. It is known that the maximum dimension of a subspace consisting entirely of entangled states in $\bigotimes_{i=1}^p\bb{C}^{d_i}$ is $d_1d_2\cdots d_p - (d_1 + d_2 + \cdots + d_p) + p - 1$ \cite{Par04,Bha06}, and it is now natural to ask whether or not there is an NPT subspace of the same dimension. Note, however, that in the multipartite setting the partial transpose can be taken on many different subsystems, and it is possible that some of a state's partial transposes are positive semidefinite while others are not. Thus the proper question to ask is for the maximum dimension of a subspace with the property that that any state with range in the subspace has at least one partial transpose that is non-positive. However, we are not aware of an answer to this question (one reason for this is that our proof of Theorem~\ref{thm:nppt_subspace} relies largely on matrix-theoretic techniques, yet quantum states are isomorphic to matrices only in the bipartite case).

The author would like to thank Jianxin Chen and John Watrous for helpful conversations. This work was supported by the Natural Sciences and Engineering Research Council of Canada.

\bibliography{../../_bibliographies_/quantum}
\end{document}